\documentclass[11pt]{article}
\usepackage{a4wide}
\usepackage{amsmath,amssymb}
\usepackage{amsthm}
\usepackage{titlesec}
\usepackage{indentfirst}

\usepackage{bbm}
\usepackage{mathrsfs}
\usepackage{amsfonts}

\usepackage{amssymb}
\usepackage{amsmath}
\usepackage{amsthm}
\usepackage{titlesec}
\usepackage{indentfirst}

\usepackage{bm}
\usepackage{xy}      
\xyoption{all}
\usepackage{graphicx}

\theoremstyle{plain}
\newtheorem{thm}{Theorem}
\newtheorem{Pp}[thm]{Proposition}

\newtheorem{Df}[thm]{Definition}
\newtheorem{Rm}[thm]{Remark}

\renewcommand\thesection{\arabic{section}.\kern -.3em}
\renewcommand{\thesubsection}{\arabic{section}.\arabic{subsection}.\kern -.5em}

\newcommand{\ign}[1] {{}}

\begin{document}
\title{{\bf Decidability of minimization of  fuzzy  automata}}
\author{\hskip 2mm  Lvzhou Li$^{1,}$\thanks{lilvzh@mail.sysu.edu.cn (L. Li).} , Daowen Qiu$^{1,2,}$\thanks{issqdw@mail.sysu.edu.cn (Corresponding author, D. Qiu).}
\\
\small{{\it $^1$Department of
Computer Science, Sun Yat-sen University, Guangzhou 510006, China}}\\
\small{{\it $^2$SQIG--Instituto de Telecomunica\c{c}\~{o}es, Departamento de Matem\'{a}tica, Instituto Superior T\'{e}cnico, }}\\
\small{{\it  Technical University of Lisbon, Av. Rovisco Pais 1049-001, Lisbon, Portugal}}\\
}\date{ }
\maketitle

\begin{abstract}
State minimization is a fundamental problem in automata theory.  The problem   is also of great importance in the study of fuzzy automata. However, most  work in the literature  considered only  state reduction of fuzzy automata, whereas the state minimization problem is almost untouched for fuzzy automata. Thus in this paper we focus on the latter problem. Formally, the decision version of the minimization problem of fuzzy automata is as follows:
\begin{itemize}
  \item Given a fuzzy  automaton $\mathcal{A}$ and a natural number $k$, that is, a pair $\langle \mathcal{A}, k\rangle$,  is there a $k$-state fuzzy  automaton equivalent to  $\mathcal{A}$?
\end{itemize}
We prove for the first time that the above problem is decidable for fuzzy automata over  totally ordered lattices.
To this end, we first give the concept of  systems of fuzzy polynomial equations and then present a procedure to solve these systems. Afterwards, we apply the solvability of  a system of fuzzy polynomial equations to the minimization problem mentioned above, obtaining the decidability. Finally,  we point out that the  above problem is  at least as hard as PSAPCE-complete.

\end{abstract}


\section{Introduction}

Fuzzy finite automata were introduced by Wee \cite{Wee67} and Santos \cite{San68,San72} in the late 1960s.
Subsequently, the fundamentals of fuzzy language theory were established by Lee and Zadeh \cite{LZ69}, and by Thomason and Marinos \cite{TM74}.
For early surveys of the theory of fuzzy finite-state machines we can refer to \cite{DP80,GSG77,KL97}.  Malik and Mordeson et al have systematically established the theory of algebraic fuzzy finite automata \cite{MM00,MM02}, and the corresponding concept of fuzzy languages was introduced in  \cite{MM00,MM02,Qiu04}.
The idea of fuzzy finite automata valued in some abstract sets may go back to Wechler's work \cite{Wec78}. Notably,
fuzzy finite automata have many significant applications such as in learning systems \cite{Wee69}, pattern recognition \cite{DP80,MM02}, data base theory \cite{DP80,MM02}, discrete event systems \cite{LY02,Qiu05}, and neural networks \cite{DP80,GSG77,MM02,GOT99,OTG98}. To date, there are still lots of important contributions to fuzzy finite automata in the academic community, but the minimization problem of fuzzy finite automata has still not been solved and remains open.

The purpose of this paper is to study the minimization problem of fuzzy finite automata, simply {\it fuzzy automata}. Minimizing a fuzzy automaton means finding a fuzzy automaton that has minimal number of states among the set of all fuzzy automata equivalent to the given one. Since fuzzy automata are generalizations of nondeterministic finite automata (NFAs), we first briefly recall the history of minimizing NFAs and related problems, which is helpful for us to understand the corresponding problems of fuzzy automata.

The study of the minimization problem for finite automata dates back to the early beginnings of automata theory. Although the problem's fundamental nature is sufficient to justify its study, the problem is of practical relevance, because regular languages are used in many applications, and one may like to represent the languages succinctly. As we know, deterministic finite automata (DFAs)  and nondeterministic finite automata (NFAs) are two important representations of regular languages.  It is well known that there exists an efficient algorithm running in time $O(n\log n)$ for minimizing any given $n$-state DFA \cite{Hop71}. The idea of the algorithm is to merge indistinguishable states of the given automaton, obtaining a unique reduced DFA which is the minimal one we find. However in the case of NFAs,
the reduced automaton obtained by merging indistinguishable states is not always minimal.  Therefore, one has to find other methods for minimizing NFAs.

It is obvious that there are only finite number of NFAs which have fewer states than the given automaton. As a result, we can always resort to an exhaustive procedure to find a minimal equivalent NFA. But this procedure is too lengthy to be practical. For that reason, many attempts have been made to search for more efficient algorithms for the minimization problem of NFAs.
The first non-exhaustive search algorithm for minimizing NFAs was proposed by Kameda and Weiner \cite{KW70} in 1970. In 1992, Sengoku \cite{Sen92} proposed another algorithm which was claimed to be more effective.  However, both of the two algorithms are not of polynomial-time complexity, which is contrary to the deterministic case.
 Indeed, the minimization problem of NFAs was proven to be computationally hard (PSPACE-complete) by Jiang and Ravikumar \cite{Jiang93}  in 1993.  Nevertheless this hardness did not prevent investigators from considering this problem any furthermore (e.g. \cite{Mel99, Mel20, Pol05,CC03}).
In addition, although  the problem is known to be PSPACE-complete, one may propose algorithms of not
too high complexities to compute relatively good approximations.

 A closely related problem to minimization is the state reduction problem  which does not aim at  a minimal NFA, but  find ``reasonably'' small NFAs that can be constructed efficiently. To this end, Ilie and Yu \cite{IY02} (see also \cite{IY03}) first proposed the  method of
reducing the size of NFAs by using equivalence relations (right invariant equivalence and left invariant equivalence).  This work leaded to a series of further work on this topic (e.g. \cite{Yu:b, Yu:c,Yu:d, Yu:a, CC04}).   The basic idea of the method  is to   merge indistinguishable states, which   resembles
the minimization algorithm for DFAs, but is more complicated. Note that right invariant equivalence
was also studied from another aspect by Calude et al. \cite{Cal00} under the name well-behaved equivalence. In fact, right invariant equivalence shares the same idea with a central concept in theoretical computer science---``bisimulation equivalence'' which originates from concurrency theory and  modal logic but now  is a rich concept  appearing in various areas of theoretical computer science such as formal verification, model checking, etc.

Fuzzy finite automata, simply fuzzy automata,  are generalizations of NFAs, and the mentioned problems concerning minimization and reduction of states are also present in work with fuzzy automata. Reduction of number of states of fuzzy automata was studied in \cite{Pee91, MS99, BG02, CM04, LL07,  Pet06, WQ10}, and the algorithms given there were also based on the idea of  merging indistinguishable states. Although these algorithms were claimed to be  minimization algorithms, the term  ``minimization'' in the mentioned references does not mean state minimization in our sense,  since it does not mean the usual construction of the minimal one in the set of all fuzzy automata recognizing a given fuzzy language, but just the procedure of  merging indistinguishable states. Therefore, these are just state reduction algorithms. Very recently,  \'{C}iri\'{c} et al \cite{JCS10} made  new  progress on  reduction of fuzzy automata. Specifically, they extended the concept of   equivalence on NFAs to  fuzzy equivalences on fuzzy automata.  Then they studied  state reduction of fuzzy  automata by means of fuzzy equivalences, which parallels the work of Ilie and Yu \cite{IY02} on NFAs. In all previous papers dealing with reduction of fuzzy automata (cf. \cite{Pee91,MS99, BG02, CM04, LL07, Pet06, WQ10}) only reductions by means of crisp equivalence were investigated, whereas \'{C}iri\'{c} et al  \cite{JCS10} showed that better reductions can be achieved by employing fuzzy equivalence.

From the above introduction, one can find that so far much attention  has been paid to  reduction of fuzzy automata, but another important problem  is almost untouched---the minimization problem.  Obviously, the minimization problem looks more difficult than the state reduction problem, similar to the  situations in the NFA case. However, this should not become the reason for one not trying to address the minimization problem. Indeed, in this paper we are going to  attempt to consider this problem and hope this would inspire some further discussion. Formally, the decision version of the minimization problem of fuzzy automata is described as follows:
\begin{itemize}
  \item Given a fuzzy  automaton $\mathcal{A}$ and a natural number $k$, that is, a pair $\langle \mathcal{A}, k\rangle$,  is there a $k$-state fuzzy automaton equivalent to  $\mathcal{A}$?
\end{itemize}

Note that fuzzy automata can be defined over different underlying structures of truth values. A general model is the one defined over  complete residuated lattices proposed by Qiu \cite{Qiu01, Qiu02}. In this paper,  fuzzy automata are restricted to be over totally ordered lattices. We prove that the above problem is decidable for this model. The result will be based, on one side, on the decidability of the equivalence problem of fuzzy automata and, on the other side, on the solvability of a system of fuzzy polynomial equations. The former base is an active problem having received much attention and can be addressed successfully. The later one is a new problem and involves a new concept---systems of fuzzy polynomial equations. A system of fuzzy polynomial equations resembles the usual notion of a system of  polynomial equations on the real-number field. However, the ideas for solving the two types of systems are greatly different.
Therefore, in this paper we will first present a procedure to solve systems of fuzzy polynomial equations, and afterwards apply the result to address the minimization problem of fuzzy automata, obtaining the decidability result.

 The reminder of this paper is organized as follows.  In Section 2 some preliminary knowledge is given.   In Section  3 we give the concept of systems of fuzzy polynomial equations and present a procedure to solve these systems, establishing a crucial base for the minimization problem. In Section 4 we give a condition for two fuzzy automata being equivalent, obtaining another base for the above problem.  Section 5 is a main part of this paper where we prove that the minimization problem of fuzzy automata is decidable and furthermore we point out that it is at least as hard as PSPACE-complete. Finally, some conclusions are made in Section 6.

\section{Preliminaries} \label{Sec-pre}
 In this paper $\mathbf{L}=(L,\wedge,\vee, 0, 1)$ means a  lattice with the least element $0$ and the greatest element $1$, where $L$ is a poset with the partial order relation $\leq$, and for any $x,y\in L$, $x\vee y$ and $x\wedge y$ denote the {\it least upper bound} and the {\it greatest lower bound } of $\{x, y\}$, respectively. $\mathbf{L}$ is said to be   {\it totally ordered}, if  $L$ is a totally ordered set.  From now on we assume that $\mathbf{L}$ is a totally ordered  lattice.

 A {\it fuzzy subset} of a set $A$ over $\mathbf{L}$, or simply a {\it fuzzy subset}
of $A$, is any mapping from $A$ into $L$.  An $n\times m$ {\it fuzzy matrix} over $\mathbf{L}$ is given by $A=[a_{ij}]_{n\times m}$ where  $a_{ij}\in L$. $L^{n\times m}$ denotes the set of all $n\times m$  fuzzy matrices over $\mathbf{L}$. Let $A\in L^{n\times m}$ and $B\in L^{m\times l}$. Their product, denoted by $A\circ B$, is defined by
\begin{equation}(A\circ B)_{ij}=\bigvee_{k=1}^m a_{ik}\wedge b_{kj}\label{Product}\end{equation}
where $1\leq i\leq n$, $1\leq j\leq l$.
  For two fuzzy matrices $A$ and $B$, their direct sum is $A\oplus B=\left(
                                                                                                       \begin{array}{cc}
                                                                                                         A & O \\
                                                                                                         O & B \\
                                                                                                       \end{array}
                                                                                                     \right)
$ where $O$ denotes a matrix with all elements being the least element $0$.

 Given $a,b\in L$ satisfying $a\leq b$, $[a,b]=\{x\in L: a\leq x\leq b\}$ is called an {\it  interval} of $L$. By this, we may identify $L$ with the interval $[0,1]$\footnote{One should not confuse the interval with the real unit interval $[0,1]$. In this paper $0$ and $1$ respectively stand for the least element and the greatest element of a lattice.}.
An $n$-dimensional {\it interval vector} is an $n$-tuple $\mathcal{X}=(X_1,\cdots, X_n)$ where $X_i\subseteq L$. Given two $n$-dimensional interval vectors ${\cal X}$ and ${\cal Y}$, their intersection is defined by $\mathcal{X}\cap \mathcal{Y}=(X_1\cap Y_1, \cdots, X_n\cap Y_n)$. Let $S_1$ and $S_2$ be two sets of interval vectors with the same dimension.
The {\it product} of $S_1$ and $S_2$ is defined by \begin{equation}S_1\star S_2=\{\mathcal{X}\cap \mathcal{Y}: \mathcal{X}\in S_1, \mathcal{Y}\in S_2\}.\end{equation} Note that two nonempty sets $S_1$ and $S_2$ may produce a product set with $(\emptyset,\emptyset,\cdots,\emptyset)$ as its unique element where $\emptyset$ denotes the empty set.
It holds that $|S_1 \star S_2|\leq |S_1|.|S_2|$  where $|S|$ denotes  the cardinality of set $S$.

\section{Solving the systems of  fuzzy polynomial equations}\label{Sec-poly}

A {\it fuzzy monomial } over $\mathbf{L}=(L,\wedge,\vee, 0, 1)$ is a finite expression $M=a\wedge x_1\wedge x_2\wedge\cdots \wedge x_k$ where $k\geq 0$, $a\in L$, $x_1,\cdots, x_k$ are variables with values in $L$. When $k=0$, it means a constant.  If the coefficient $a$ is the greatest element $1$, the monomial is then simply written by $x_1\wedge x_2\wedge\cdots \wedge x_k$. A {\it fuzzy polynomial } over $\mathbf{L}$ is an expression of the form $M_1\vee M_2\vee \cdots \vee M_k$ where $k\geq 1$ and $M_1,\cdots, M_k$ are monomials over $\mathbf{L}$. It is usually denoted by $P=\bigvee_{i=1}^{k} M_i$.

A {\it system of  fuzzy  polynomial (in)equations (SFPI)} is of the form
\begin{equation}
  \begin{cases}
    P_1 \bowtie a_1\\
~~~~~\vdots\\
P_m \bowtie a_m
  \end{cases}\label{SFI}
\end{equation}
where $\bowtie\in \{=, \leq, \geq,<, > \}$, $P_i$'s are  polynomials and $a_i\in L$ for $i=1,\cdots, m$.  If $\bowtie$ is always chosen to be ``$=$'',  then (\ref{SFI}) is called a {\it system of  fuzzy polynomial equations (SFPE)}.  Suppose the system involves $n$ variables $(x_1,\cdots, x_n)$. $X^0=(x^0_1,\cdots, x^0_n)$ is called  a {\it  point solution } of (\ref{SFI}), if  (\ref{SFI}) holds when  $(x_1,\cdots, x_n)$ is assigned with $X^0$.  If (\ref{SFI}) has a point solution, then it is said to be {\it solvable}. Otherwise, it is {\it unsolvable}.
 An $n$-dimensional interval vector $\mathcal{X}=(X_1,\cdots, X_n)$ where $X_i\subseteq L$ is called an {\it interval solution} of (\ref{SFI}), if each $n$-tuple $(x^0_1,\cdots, x^0_n)$ with $x^0_i\in X_i$ is a point solution of  (\ref{SFI}) and $(X_1,\cdots, X_n)$ is maximum\footnote{Note that the inclusion ``$\subseteq$'' is a partial order on $2^L$, and it can  be  extended up to be a partial order on the set of interval vectors.} with respect to this property.

\begin{Rm}
   First we mention  that  a special case of the above discussion that each monomial $M$ involves only one variable was considered in \cite{Pee92}. Secondly, in the above definitions if $\vee$ and $\wedge$ are respectively replaced with  the usual addition ``$+$'' and the multiplication ``$.$'', and $L$ is replaced with the real number field $\mathbb{R}$, then $P$ is the usual  polynomial and (\ref{SFI}) is a system of polynomial (in)equations. Note that deciding the solvability of a  system of polynomial (in)equations is a very important problem which has wide applications. A more general problem is the decision problem for the existential theory of
the reals  which is the problem of deciding if the set $\{x\in \mathbb{R}^n: P(x)\}$ is non-empty, where $P(x)$ is a predicate which is a boolean combination  of atomic predicates either of
the form $f_i(x)\geq 0$ or $f_j(x)> 0$, $f$'s being real polynomials (with rational coefficients). It is known that this problem can be decided in PSPACE \cite{Ren88,canny:88} and it has been applied to the  complexity analysis on the reachability of Recursive Markov Chains \cite{EY09} and on the minimization problem of quantum and probabilistic automata \cite{MQL12}.

   \end{Rm}

Now we consider the decision problem of SFPI, that is:
\begin{itemize}
  \item  Is it decidable whether   a given system like (\ref{SFI}) is solvable or not?
\end{itemize}
In the following we consider a special case that the relation operator $\bowtie$ in (\ref{SFI}) is always chosen to be ``$=$'' and all polynomials  take $1$ as coefficients. We will show that the above problem is decidable in this case. Afterwards, we apply the result to the minimization problem of fuzzy automata.
 Hereafter our aim is to determine whether the following system is solvable or not:
\begin{equation}
  \begin{cases}
    P_1 = a_1\\
~~~~~\vdots\\
P_m = a_m
  \end{cases}\label{SFE}
\end{equation}
where all polynomials $P_i$'s take $1$ as coefficients.

First we consider the case of one equation with a monomial, i.e.
\begin{equation}
x_1\wedge x_2\wedge\cdots \wedge x_n=a.\label{moe}
\end{equation}
It is easy to see that a solution of this equation has the property that one $x_i$ equals $a$ and all others are not less than $a$. More formally, the equation has $n$ interval solutions as follows
\begin{equation}\begin{array}{c}
([a,a], [a,1],[a,1],\cdots, [a,1]),\\
([a,1], [a,a],[a,1],\cdots, [a,1]),\\
\vdots \\
([a,1], [a,1],\cdots,[a,1], [a,a])
 \end{array}\label{Int1}\end{equation}
where the interval $[a, a]$ stands for the point $a$.
Similarly, the solution of  the inequality
\begin{equation}
x_1\wedge x_2\wedge\cdots \wedge x_n\leq a \label{moi}
\end{equation}
has the property that one $x_i$ is no more than $a$ and all others can be assigned with arbitrary values. Thus, the inequality has $n$ interval solutions given by
\begin{equation}\begin{array}{c}
                  ([0,a], [0,1],[0,1],\cdots, [0,1]), \\
                   ([0,1], [0,a],[0,1],\cdots, [0,1]), \\
                    \vdots \\
                   ([0,1], [0,1],\cdots,[0,1], [0,a]).
                 \end{array}\label{Int2}
\end{equation}

Next we consider  an equation with a polynomial:
\begin{equation}
P=M_1\vee M_2\vee\cdots \vee M_k=a. \label{poly}
\end{equation}
where $M_k$'s are monomials. Eq. (\ref{poly}) is solvable if and only if one monomial  $M_i$ equals $a$ and others are no more than $a$. Therefore there are $k$ possible cases. For example, one case is as follows:
\begin{equation}
  \begin{cases}
    M_1 = a\\
  M_2 \leq a\\
~~~~~~\vdots\\
M_k \leq a
  \end{cases}\label{SFE2}
\end{equation}
 We can search for all possible interval solutions for each  equation or inequality in (\ref{SFE2}) as we dealt with (\ref{moe}) and (\ref{moi}).

Let $n_i$ be the number of variables in $M_i$ and let $n$ be the total number of variables in (\ref{SFE2}). Then $n_i\leq n$ for $i=1,\cdots, k$.
 Here we need extend a solution on $n_i$ variables to one on $n$ variables. For example, the first equation in (\ref{SFE2}) has only $n_1$ variables and thus its interval solution is an $n_1$-tuple $\mathcal{X}=(X_1,\cdots, X_{n_1})$ with $X_i\subseteq L$. Assume that all variables in (\ref{SFE2}) are ordered by some order, for example the lexicographic order. Without loss of generality assume that the variables involved in $M_1$ are ordered at the head.
 Then the  variables that appear in (\ref{SFE2}) but not in  $M_1$  have no impact on the first equation  in (\ref{SFE2}) and can be chosen arbitrarily from $L$. Therefore, $\mathcal{X}$ can be extended to the following form
\begin{equation}
\overline{\mathcal{X}}=(X_1,\cdots, X_{n_1}, [0,1],\cdots,[0,1]).
\end{equation}
By this way, an interval solution for each equation or inequality of (\ref{SFE2}) will be an $n$-dimensional interval vector.
Let $S_i$ be the set of interval solutions of the $i$th formula in (\ref{SFE2}).
Then the  system (\ref{SFE2}) has a set of interval solutions  $S={S_1}\star\cdots \star {S_k}$ with $|S|\leq n_1.n_2\cdots n_k\leq n^k$.

Note that the equality (\ref{poly}) has $k$ cases to be verified. Therefore, the number of interval solutions of Eq. (\ref{poly}) is upper bounded by $kn^k$.

Now we return to the solution of system (\ref{SFE}). Let the $i$th equation have a set of interval solutions $T_i$ where $i=1,\cdots, m$. Then the interval solution set of (\ref{SFE}) is $T=T_1\star\cdots \star T_m$. Assume that the system (\ref{SFE}) involves $n$ variables and each polynomial consists of at most $k$ monomials. Then we have $|T|\leq |T_1|\cdots | T_m|\leq (kn^k)^m$.
Consequently, we obtain the following theorem.
\begin{thm}
There exists an algorithm running in time $O((kn^k)^m)$ to search for all possible solutions to a system of fuzzy polynomial equations like (\ref{SFE}) and thus it is decidable whether the system is solvable, where $n$ is the number of variables, $m$ is  the number of equations, and $k$ is the largest number of monomials involved in an equation.\label{thm-sfe1}
\end{thm}

 \begin{proof} Using the procedure described above, we can exhaustively search for all possible solutions of the system. The number of interval solutions is at most $(kn^k)^m$. So the time is $O((kn^k)^m)$. Note that the system is solvable if and only if there exists an interval solution in the form $(X_1,\cdots, X_n)$ where $X_i\neq\emptyset$ for all $i$. Thus, after finding all possible interval solutions, we check if there exists a solution satisfying the above condition. If yes, then the system is solvable. Otherwise, it is unsolvable.
\end{proof}

\begin{Rm} In the above discussion, it was assumed that  a system consists of only equations and all involved polynomials take $1$ as coefficients. In fact, the above procedure can be extended to more general cases. Firstly, it can be easily adopted to deal with systems of inequalities and equations with $1$ as coefficients. Secondly, by the similar idea, we can also deal with the most general case---systems consist of inequalities and equations with arbitrary coefficients. However, we will not discuss these general cases in any more detail, since they have no close relation  with our main purpose in this paper---minimization of fuzzy finite automata.\end{Rm}

Theorem \ref{thm-sfe1} has presented a procedure to search for all solutions to  systems like (\ref{SFE}). However, if we only want to determine whether the system is solvable or not, but not to search for all solutions, then we can present a simpler procedure for that. Let ${\cal V}$ denote the set consisting of distinct elements from $\{a_1, \cdots, a_m\}$.
Then we get the following result.
\begin{thm}\label{thm-sfe2}The system (\ref{SFE}) is solvable if and only if there exists  $X^0\in \mathcal{V}^n$ such that $X^0$ is a point solution of (\ref{SFE}) where $n$ is the number of variables involved in the system.
\end{thm}
\begin{proof} First  the ``if'' part is ready. Thus we focus on the ``only if'' part. Suppose that the system (\ref{SFE}) is solvable. Then it necessarily has an interval solution $(X_1,\cdots, X_n)$ where $X_i\neq\emptyset$ for all $i$. More specifically, since every interval $X_i$ is the intersection of intervals like the ones given in (\ref{Int1}) and (\ref{Int2}), every interval $X_i$ is finally among the following forms $[a,b], [0,c], [d,1], [0,1]$ where $a,b,c,d\in{\cal V}$. Thus, we  get a point solution $(x^0_1,\cdots, x^0_n)$ where each $x^0_i$ is given by
\begin{equation}x_i=\left\{
  \begin{array}{ll}
    a \text{~or~} b, & \hbox{if $X_i=[a,b]$;} \\
    c, & \hbox{if $X_i=[0,c]$;} \\
    d, & \hbox{if $X_i=[d,1]$;} \\
    y\in \mathcal{V}, & \hbox{if $X_i=[0,1]$}
  \end{array}
\right.\end{equation}
where $y$ is an arbitrary value from $\mathcal{V}$.
It is obvious that $(x^0_1,\cdots, x^0_n)\in \mathcal{V}^n$. Hence, we have completed the proof.
\end{proof}
Having in mind Theorem \ref{thm-sfe2} we propose the following algorithm (Algorithm I) for establishing the solvability of  systems like (\ref{SFE}).
\begin{center}

 \fbox{\parbox{10cm}{
 {\small\vskip 1mm
{\bf Input:} a system like (\ref{SFE}) with $n$ variables\\
{\bf output:} give a solution or return ``unsolvable''\\
{\bf Step 1:} \begin{quote} find the set $\mathcal{V}$ from $\{a_1, \cdots, a_m\}$.\end{quote}
{\bf Step 2:}  \begin{quote} {\bf While} (take an element $X^0$ from $\mathcal{V}^n$) {\bf do}\\
 If ($X^0$ is a solution of (\ref{SFE})) return $X^0$
\end{quote}
{\bf Step 3:}
 \begin{quote} return  ``unsolvable''
\end{quote}

}}}

 Algorithm I: determine whether a system of fuzzy polynomial equations like (\ref{SFE}) is solvable. \vskip 1mm
\end{center}

In Section 5, we will apply the above results to address the minimization problem of fuzzy  automata.

\section{Equivalence of fuzzy  automata}\label{Sec-eqv}
Below we first give the definition of fuzzy  automata and then present a condition for two fuzzy  automata being equivalent which is an indispensable base for the minimization problem. Some notations used in the sequel are firstly explained here.
 Let $\Sigma^*$ denote the free monoid over the alphabet $\Sigma$, and let $\lambda\in\Sigma^*$ be the empty word. For $x\in \Sigma^*$, $|x|$ denotes its length. Let $\Sigma^{\leq k}=\{x\in\Sigma^*: |x|\leq k\}$.

\begin{Df}\label{df-ffa}A fuzzy finite automaton over  $\mathbf{L}=(L,\wedge,\vee, 0, 1)$, simply  fuzzy automaton, is a five-tuple ${\cal A}=(S,\Sigma, \pi,\delta, \eta)$ where $S$ is a finite state set, $\Sigma$ is a finite alphabet, $\pi$  is a fuzzy subset of $S$, called the fuzzy set of initial states, $\delta: S\times \Sigma \times S\rightarrow L$ is a fuzzy subset of $S\times \Sigma \times S$, called  the fuzzy  transition function, $\eta$  is a fuzzy subset of $S$, called the fuzzy set of terminal states.
\end{Df}

We can interpret $\delta(s, \sigma, s')$ as the degree to which an input letter $\sigma\in\Sigma $ causes a transition from a state $s\in S$ into a state $s'\in S$.
 The mapping $\delta$ can be extended
up to a mapping $\delta^*: S\times \Sigma^* \times S\rightarrow L$ as follows:
\begin{equation}
\delta^*(s,\lambda, s')=\left\{
                         \begin{array}{ll}
                           1, & \hbox{if $s=s'$,} \\
                           0, & \hbox{otherwise}
                         \end{array}
                       \right.
\end{equation}
 where $s,s'\in S$, and
 \begin{equation}
\delta^*(s,x\sigma, s')=\bigvee_{t\in S} \delta^*(s,x, t)\wedge\delta(t,\sigma,s')
 \end{equation}
where $s,s'\in S$, $x\in \Sigma^*$ and $\sigma\in\Sigma$.

A {\it fuzzy language} in $\Sigma^*$ over $\mathbf{L}$, or briefly a fuzzy language, is a fuzzy subset of $\Sigma^*$, i.e., a mapping from $\Sigma^*$ into $L$. $\mathfrak{F}(\Sigma^*)$ stands for the set of all fuzzy languages in  $\Sigma^*$. A fuzzy automaton ${\cal A}=(S,\Sigma, \pi,\delta, \eta)$  {\it recognizes} a fuzzy language $f\in \mathfrak{F}(\Sigma^*)$ if for any $x\in \Sigma^*$
 \begin{equation}f(x)=\bigvee_{s,t\in S} \pi(s)\wedge \delta^*(s,x,t)\wedge \eta(t)\label{language}.\end{equation}
 In other words, the equality (\ref{language}) means that the membership degree of the word $x$ belonging to the fuzzy language $f$ is equal to the
degree to which ${\cal A}$ recognizes or accepts the word $x$. The unique fuzzy language recognized by a fuzzy automaton ${\cal A}$ is denoted by $f_{\cal A}$.

Let $|S|=n$. For each word $x\in\Sigma^*$, we define an $n\times n$ fuzzy matrix $\delta_x$ by
\begin{equation}\delta_x(s,s')=\delta^*(s,x,s')
\end{equation}
where $s,s'\in S$. Then it is readily seen that
\begin{equation}\delta_x=\delta_{x_1}\circ \cdots \circ \delta_{x_k}
\end{equation}
for $x=x_1\cdots x_k\in\Sigma^*$. By using this notation,  the equality (\ref{language}) can be rewritten as
\begin{equation}f(x)=\pi\circ \delta_x\circ \eta\end{equation}
where and hereafter we identify the fuzzy subsets $\pi, \eta$  with a $1\times n$ and an $n\times 1$ fuzzy matrices, respectively.

A basic problem concerning fuzzy automata is the equivalence problem, that is, deciding whether two given fuzzy automata recognize the same fuzzy language.    Formally the definition is given as follows.
\begin{Df}
Two fuzzy automata ${\cal A}_1$ and ${\cal A}_2$ are said to be $k$-equivalent, denoted by ${\cal A}_1\approx_k {\cal A}_2$, if $f_{{\cal A}_1}(x)=f_{{\cal A}_2}(x)$ holds for all $x\in\Sigma^{\leq k}$. Furthermore, they are said to be equivalent, denoted by ${\cal A}_1\approx {\cal A}_2$, if  $f_{{\cal A}_1}(x)=f_{{\cal A}_2}(x)$ holds for all $x\in\Sigma^*$.
\end{Df}
In the following, we give a condition for two fuzzy automata being equivalent.
\begin{thm}\label{thm-eqv} Two  fuzzy automata  ${\cal A}_i=(S^{(i)},\Sigma, \pi^{(i)}, \delta^{(i)},\eta^{(i)})$ $(i=1,2)$ are equivalent if and only if they are $(d^{(n_1+n_2)}-1)$-equivalent, where $n_1$ and $n_2$ are numbers of states of ${\cal A}_1$ and ${\cal A}_2$, respectively, and $d$ is the cardinality of the following set:
\begin{eqnarray*}V&=&\{\delta^{(1)}(s,\sigma,t): s,t\in S^{(1)},\sigma\in\Sigma\}\cup\{\eta^{(1)}(s): s\in S^{(1)}\}\\
&\cup&\{\delta^{(2)}(s,\sigma,t): s,t\in S^{(2)},\sigma\in\Sigma\}\cup\{\eta^{(2)}(s): s\in S^{(2)}\}.\end{eqnarray*}
\end{thm}
\begin{proof}Let $M(\sigma)=\delta^{(1)}_{\sigma}\oplus \delta^{(2)}_{\sigma}$ for all $\sigma\in\Sigma$, and $\eta=\left(
                                                                                                          \begin{array}{c}
                                                                                                            \eta^{(1)} \\
                                                                                                            \eta^{(2)} \\
                                                                                                          \end{array}
                                                                                                        \right)
$, being an $(n_1+n_2)\times 1$ matrix.  Then we have
\begin{eqnarray}
& &f_{{\cal A}_1}(x)=\pi^{(1)}\circ \delta^{(1)}_x\circ \eta^{(1)}=(\pi^{(1)},\mathbf{0})\circ M(x)\circ\eta,\\
& &f_{{\cal A}_2}(x)=\pi^{(2)}\circ \delta^{(2)}_x\circ \eta^{(2)}=(\mathbf{0},\pi^{(2)})\circ M(x)\circ\eta
\end{eqnarray}
where $M(x)=M(x_1)\circ\cdots \circ M(x_k)$ for  $x=x_1\cdots x_k\in \Sigma^*$ and $\mathbf{0}$ denotes a row vector with all elements being $0$. Let $$\varphi(k)=\{M(x)\circ\eta: x\in \Sigma^{\leq k}\}$$ and $$\varphi=\{M(x)\circ\eta: x\in \Sigma^*\}.$$ Then  the following result holds.
\begin{Pp}There exists an integer $l\leq d^{n_1+n_2}-1$ such that $$\varphi(0)\subset \varphi(1)\subset \cdots \subset \varphi(l)=\varphi(l+1)=\varphi(l+2)= \cdots=\varphi$$ \label{Prop1}
where $\subset$ denotes the proper inclusion.\end{Pp}
{\noindent\bf Proof of Proposition \ref{Prop1}}. First it is obvious that $\varphi(i)\subseteq\varphi(i+1)\subseteq \varphi$ holds for $i=0,1,\cdots$. Note that $|\varphi|\leq d^{n_1+n_2}$. Thus, there exists $l\leq|\varphi|-1\leq d^{n_1+n_2}-1$ such that $\varphi(l)=\varphi(l+1)$. Next we prove that $\varphi(l)=\varphi(l+j)$ holds for $j=2,3,\cdots$. In fact we only need prove $\varphi(l)=\varphi(l+2)$. Take $\upsilon\in \varphi(l+2)$. We have
 \begin{eqnarray*}v&=&M(x)\circ\eta, ~~~~~~~~~~~~~~x\in \Sigma^{\leq(l+2)} \\
&=&M(\sigma)\circ M(x')\circ\eta, ~~~~~\sigma\in\Sigma, x'\in\Sigma^{\leq(l+1)} \\
&=&M(\sigma)\circ v',~~~~~~~~~~~~~v'\in\varphi(l)\\
&\in&\varphi(l+1)=\varphi(l).
\end{eqnarray*}
This completes the proof of Proposition \ref{Prop1}.

By the above results  we obtain
\begin{eqnarray*}
{\cal A}_1\approx {\cal A}_2&\Leftrightarrow&f_{{\cal A}_1}(x)=f_{{\cal A}_2}(x),\forall x\in\Sigma^*\\
&\Leftrightarrow& (\pi^{(1)},\mathbf{0})\circ M(x)\circ\eta=(\mathbf{0},\pi^{(2)})\circ M(x)\circ\eta, \forall x\in\Sigma^*\\
&\Leftrightarrow& (\pi^{(1)},\mathbf{0})\circ v=(\mathbf{0},\pi^{(2)})\circ v, \forall v\in\varphi\\
&\Leftrightarrow& (\pi^{(1)},\mathbf{0})\circ v'=(\mathbf{0},\pi^{(2)})\circ v', \forall v'\in\varphi(k)\\
&\Leftrightarrow& f_{{\cal A}_1}(x)=f_{{\cal A}_2}(x), \forall x\in\Sigma^{\leq k}\\
&\Leftrightarrow&{\cal A}_1\approx_k {\cal A}_2
\end{eqnarray*}
where $k=d^{n_1+n_2}-1$. Thus we have completed the proof of Theorem \ref{thm-eqv}.\end{proof}
\section{Minimization of fuzzy  automata}\label{Sec-min}
The study of the minimization problem for finite automata dates back to the early beginnings of automata theory. This problem is  of practical relevance, because
regular languages are used in many applications, and one may like to represent the
languages succinctly. In the past over forty years, fuzzy automata have received much attention.  In the study of fuzzy automata, the minimization problem is also a very important problem  and worthy of serious consideration.  Note that a closely related problem to minimization is the state reduction problem which does not aim at a minimal one, but find ``reasonable'' small  automata which can be efficiently constructed. As mentioned in Section 1, while much attention has been paid to  reduction of fuzzy automata, the minimization problem is almost untouched. Thus in this section we focus on the minimization problem.
 Formally, the decision version of the minimization problem of fuzzy automata,  is  as follows:
\begin{itemize}
  \item Given a fuzzy automaton $\mathcal{A}$ and a natural number $k$, that is, a pair $\langle \mathcal{A}, k\rangle$,  is there a $k$-state fuzzy automaton equivalent to  $\mathcal{A}$?
\end{itemize}

In this section, we  prove the above problem  is decidable, based on the solvability of systems of fuzzy polynomial equations (Theorem \ref{thm-sfe2}) and on the condition of equivalence between fuzzy automata (Theorem \ref{thm-eqv}). Furthemore, we point out that the  minimization problem is at least as hard as PSPACE-complete.
\subsection{Decidability of minimization}

  The main idea for addressing the  minimization problem of fuzzy automata  is briefly depicted as follows:

\begin{enumerate}
  \item  First, for a given pair $\langle \mathcal{A}, k\rangle$, define the set
\begin{equation}
\mathbb{S}(\mathcal{A}, k)=\{{\cal A}':  {\cal A}'\approx {\cal A}, \text{~is a fuzzy automaton with~} k \text{~states} \}. \label{Set}
\end{equation}
  \item Next,   we  prove that $\mathbb{S}(\mathcal{A}, k)$  can be described by a system of  fuzzy polynomial equations.  Then, by Theorem \ref{thm-sfe2}, there exists an algorithm to decide whether $\mathbb{S}(\mathcal{A}, k)$ is nonempty or not, and furthermore, if it is nonempty, a $k$-state fuzzy automaton $\mathcal{A}'$ equivalent to $\mathcal{A}$ is given.
  \end{enumerate}

Specifically, we have the following result.
\begin{thm}The minimization problem of fuzzy automata is decidable. \label{thm-min}
\end{thm}
\begin{proof} For a given $n$-state fuzzy automaton ${\cal A}=(S,\Sigma, \pi, \delta, \eta)$  and  a natural number $k$, define the set $\mathbb{S}(\mathcal{A}, k)$ given by (\ref{Set}). Suppose that a $k$-state fuzzy automaton has the form ${\cal A}'=(S',\Sigma, \pi', \delta', \eta')$ with $|S'|=k$. Equivalently, $\mathbb{S}(\mathcal{A}, k)$  can be represented as follows:
\begin{equation}
\mathbb{S}(\mathcal{A}, k)=\{X\in L^{2k+|\Sigma|k^2}: P(X)\}\label{set}
\end{equation}
where:
\begin{itemize}
  \item  $X$ is a vector consisting of  $(2k+|\Sigma|k^2)$ variables which take values from $L$. More specifically, $X$ is decomposed into several parts:
\begin{itemize}
  \item  $X_{\pi'},X_{\eta'}\in L^k$  are respectively used to represent the initial state $\pi'$ and the final state $\eta'$ of ${\cal A}'$.
  \item $X_{\delta_\sigma}\in L^{k^2}$ for each $\sigma\in \Sigma$ is used to represent a fuzzy transition matrix $\delta_\sigma$ of ${\cal A}'$.
 \end{itemize}
  \item  $P(X)$ is a system of fuzzy polynomial equations given by:
\begin{equation}
X_{\pi'}\circ X_{\delta_x}\circ X_{\eta'}=\pi\circ \delta_x \circ\eta, \forall x\in \Sigma^* \label{SE1}
\end{equation}
where $X_{\delta_x}=X_{\delta_{x_1}}\circ\cdots \circ X_{\delta_{x_m}}$ for $x=x_1\cdots x_m\in\Sigma^*$.\end{itemize}

It is obvious that a point $X^0\in L^{2k+|\Sigma|k^2}$ is a solution to  the system (\ref{SE1}) if and only if it represents a $k$-state fuzzy automaton ${\cal A}'$ equivalent to ${\cal A}$.
However, since (\ref{SE1}) consists of infinitely many equations, there exists no algorithmic procedure to determine whether it is solvable. Thereby  we need reduce (\ref{SE1}) to a  system of finitely many equations.

Let
\begin{equation*}V=\{\pi(s), \eta(s), \delta_\sigma(s,s'): s,s'\in S, \sigma\in\Sigma\}\end{equation*}
and let $$d=|V|+|\Sigma|k^2+k.$$  We consider the following finite system of equations:
\begin{equation}
X_{\pi'}\circ X_{\delta_x}\circ X_{\eta'}=\pi\circ \delta_x \circ\eta, \forall x\in \Sigma^*  \text{~with~} |x|\leq d^{n+k}-1.  \label{SE2}
\end{equation}
Indeed, we obtain the following result.

\noindent{\bf Claim 1.} {\it $X^0\in L^{2k+|\Sigma|k^2}$ is a solution to  the system (\ref{SE1}) if and only if it is a solution to the system (\ref{SE2}).}
\begin{proof}
The ``only if '' part is obvious. Thus we prove the ``if'' part. Suppose that  $X^0\in L^{2k+|\Sigma|k^2}$  is a solution to the system (\ref{SE2}). Assume that the number of different values involved by $\{X^0_{\eta'}, \eta, \delta_\sigma, X^0_{\delta_\sigma}: \sigma\in\Sigma\}$ is $d'$. Then we have $d'\leq d$. Now it follows from Theorem \ref{thm-eqv} that if $X^0_{\pi'}\circ X^0_{\delta_x}\circ X^0_{\eta'}=\pi\circ \delta_x \circ\eta$ holds for $x\in\Sigma^*$ with $|x|\leq (d')^{n+k}-1$, then it also holds for all $x\in \Sigma^*$. Remember that $X^0$ is a solution to the system (\ref{SE2}) and $d'\leq d$. Thus by the above discussion $X^0$ is also a solution to the system (\ref{SE1}). This completes the proof of Claim 1.
\end{proof}

Furthermore, we can reduce the system (\ref{SE2}) to a more simple one. Formally, we have

\noindent{\bf Claim 2.} {\it The system (\ref{SE2}) is solvable if and only if the following system is solvable:}
\begin{equation}X_{\pi'}\circ X_{\delta_x}\circ X_{\eta'}=\pi\circ \delta_x \circ\eta, \forall x\in \Sigma^*  \text{~with~} |x|\leq |V|^{n+k}-1.  \label{SE3}\end{equation}
\begin{proof} First, the ``only if'' part is obvious, since $|V|<d$. Now suppose that the system (\ref{SE3}) is solvable.  Note that the values of the right side of (\ref{SE3})  are necessarily to be in the set $V$.
Therefore, by Theorem 4, the system (\ref{SE3}) is solvable in the set $V$, which means that there exists $X^0\in V^{2k+|\Sigma|k^2}$ such that $X^0_{\pi'}\circ X^0_{\delta_x}\circ X^0_{\eta'}=\pi\circ \delta_x \circ\eta$ holds for $x\in\Sigma^*$ with $|x|\leq |V|^{n+k}-1$. Then, by Theorem \ref{thm-eqv}, the equality holds  for all $x\in \Sigma^*$, and thus it necessarily holds for $x\in \Sigma^*$ with $|x|\leq d^{n+k}-1$, which means that $X^0$ is a solution to the system  (\ref{SE2}).
\end{proof}

 At length, $P(X)$ has been reduced to be a finite one given by (\ref{SE3}). Correspondingly, determining whether the set $\mathbb{S}(\mathcal{A}, k)$ is nonempty is equivalent to determining whether the system (\ref{SE3}) is solvable. The later problem can be addressed in terms of Theorem \ref{thm-sfe2}.

In conclusion the above procedure is refined to  a minimization  algorithm  (Algorithm II). This completes the proof of Theorem \ref{thm-min}.
\end{proof}
\begin{center}

   \fbox{\parbox{12cm}{
 {\small\vskip 1mm
{\bf Input:} a fuzzy automaton ${\cal
A}$ and a natural number $k$\\
{\bf output:} either return a fuzzy automaton ${\cal
A}'$  equivalent to ${\cal
A}$ with $k$ states or return ``empty''\\
{\bf Step 1:}  \begin{quote}construct the set $\mathbb{S}(\mathcal{A}, k)$ given in (\ref{set}) where $P(x)$ is given by (\ref{SE3})\end{quote}
{\bf Step 2:}  \begin{quote}invoke Algorithm I to decide whether  $\mathbb{S}(\mathcal{A}, k)$  is nonempty\end{quote}
{\bf Step 3:}
  \begin{quote}{\bf If} $\mathbb{S}(\mathcal{A}, k)=\emptyset$ return ``empty''; {\bf else} return $X^0\in\mathbb{S}(\mathcal{A}, k)$\end{quote}

}}}

Algorithm II: the minimization algorithm \vskip 1mm
\end{center}

{\noindent\bf Complexity analysis.} Here we take a complexity analysis of the above procedure. In the following we assume that  the two
operations $\vee$ and $\wedge$ can be done in constant time.
For simplicity, let $c= |V|^{n+k}-1$. The number of equations in (\ref{SE3}) is
\[N=1+|\Sigma|+|\Sigma|^2+\cdots+ |\Sigma|^c.\]
To solve the system (\ref{SE3}) we need to compute the values in the right side which can be  efficiently computed  by the following two steps:
\begin{itemize}
  \item first compute the set $\varphi(c)=\{\delta_x\circ\eta: |x|\leq c\}$;
  \item then compute the value $\pi\circ v$ for each $v\in\varphi(c)$.
\end{itemize}
Note that in this first step upon computation of $\delta_x\circ\eta$, computation of $\delta_\sigma\circ \delta_x\circ\eta$ requires $O(n^2)$ operations. Thus, we compute all  $\delta_x\circ\eta$ from $l=0$ to $l=c$ where $l=|x|$. Then the number of operations required is give as follows:
\[
\begin{tabular}{rlcrrrcc}
length & number of operations\\
$l=0$    &$O(1)$\\
$l=1$    &$O(|\Sigma|n^2)$\\
$l=2$    &$O(|\Sigma|^2n^2)$\\
$\cdots$ &$\cdots$\\
$l=c$    &$O(|\Sigma|^cn^2)$\\
\end{tabular}\]
Thus, the  number of  operations required in the first step is $O(Nn^2)$. The second step requires $O(Nn)$ operations. Therefore the computation of the right side of (\ref{SE3}) requires $O(Nn^2)$  operations.

Now in order to determine whether the system (\ref{SE3}) is solvable or not, by Theorem  \ref{thm-sfe2} we need  to check all $X\in V^{2k+|\Sigma|k^2}$. For each $X$, it requires $O(Nk^2)$ operations to compute the left side.

Therefore, the total number of operations required is
\[ O(|V|^{2k+|\Sigma|k^2}Nk^2)+O(Nn^2).\]
\subsection{The minimization problem is at least as hard as PSPACE-complete}
In the last section we have presented a procedure to address the minimization problem of fuzzy automata,  but the complexity  is not good enough. Then  one may ask whether there exists a more efficient algorithm  to the minimization problem. The answer may be yes, but at present what we are sure is that  the minimization problem is at least as hard as PSPACE-complete.

We recall that in Definition \ref{df-ffa},  a fuzzy automaton ${\cal A}=(S,\Sigma, \pi, \delta, \eta)$ is defined on a general totally ordered lattice $\mathbf{L}=(L, \vee, \wedge, 0,1)$. Now we consider a special case where $L=\{0,1\}$. In this case, ${\cal A}$ reduces to an NFA, which follows from the following observations:
\begin{itemize}
  \item $\pi(s_i)=1$ stands for $s_i\in S$  being an initial state, and $\eta(s_i)=1$ indicates $s_i$ being an accepting state.
  \item  $\delta(s_i,\sigma, s_j)=1$ stands for that the current state $s_i$ will change to $s_j$ when the automaton is fed with the symbol $\sigma$.
  \item The language accepted by ${\cal A}$ is $\mathcal{L}=\{x\in\Sigma^*: f_{\cal A}(x)=1\}$.
\end{itemize}
The minimization problem of NFA attracted much attention from the academic community, for example \cite{KW70, Sen92, Jiang93, Mel99, Mel20, Pol05,CC03}. It has been proved that the minimization problem of NFA is PSPACE-complete \cite{Jiang93}.  On can refer to a survey of finite automata \cite{HK11} for learning more information. Therefore, the minimization problem of  fuzzy automata is at least as hard as PSPACE-complete.

\section{Conclusions}
In this paper we have proved for the first time that the minimization problem of fuzzy automata over totally order lattices is decidable. This result depends heavily on, the solvability of a system of fuzzy polynomial equations, a result established in this paper. Two problems worthy of further consideration are listed as follows:
\begin{itemize}
  \item Are there more efficient algorithms for the minimization problem? Although the problem is known to be at least as hard as PSPACE-complete, it is still interesting to search for more efficient algorithms for the problem as in the case of NFAs.
  \item In this paper we have  considered only minimizing fuzzy automata over totally ordered lattices. In the further study, one can consider the problem for more general cases, for example, for fuzzy automata over complete residuated lattices.
\end{itemize}

\subsubsection*{Acknowledgments}
This work is supported in part by the National
Natural Science Foundation (Nos. 61272058, 61073054, 60873055, 61100001), the Natural
Science Foundation of Guangdong Province of China (No.
10251027501000004),  the Fundamental Research Funds for the Central Universities (No. 11lgpy36), the Specialized Research Fund for the Doctoral Program of Higher Education of China
(Nos. 20100171110042, 20100171120051).

\end{document}